\documentclass{scrartcl}

\usepackage{hyperref}
\usepackage{amsmath,amssymb}
\usepackage{amsthm}
\usepackage[utf8]{inputenc}
\usepackage{bbm}

\newtheorem{theorem}{Theorem} 
\newtheorem{lemma}{Lemma} 
\newtheorem{remark}{Remark}

\numberwithin{equation}{section}

\newcommand{\diag}{\text{diag}}
\newcommand{\trace}{\text{Tr}\,}

\newcommand{\E}{\mathbb{E}}

\newcommand{\innerproduct}[1]{ \langle #1 \rangle_{HS} }

\newcommand{\posdef}{\mathcal{S}^n_{+}}
\newcommand{\xmu}{\mathcal{X}_{\mu}}
\newcommand{\mmu}{\mathcal{M}_{\mu}}
\newcommand{\diagonal}[1]{\diag \left(  #1 \right) }
\newcommand{\mualphaone}{\mathcal{M}_{1,\mu, \alpha} }
\newcommand{\mualphatwo}{\mathcal{M}_{2,\mu, \alpha} }

\begin{document}
	\title{Complex phase retrieval from subgaussian measurements}
	\author{Felix Krahmer, Dominik St\"oger}
	\maketitle


\begin{abstract}
Phase retrieval refers to the problem of reconstructing an unknown vector $x_0 \in \mathbb{C}^n$ or $x_0 \in \mathbb{R}^n $ from $m$ measurements of the form $y_i = \big\vert \langle \xi^{\left(i\right)}, x_0 \rangle \big\vert^2 $, where $ \left\{ \xi^{\left(i\right)} \right\}^m_{i=1} \subset \mathbb{C}^m  $ are known measurement vectors. While Gaussian measurements allow for recovery of arbitrary signals provided the number of measurements scales at least linearly in the number of dimensions, it has been shown that ambiguities may arise for certain other classes of measurements  $ \left\{ \xi^{\left(i\right)} \right\}^{m}_{i=1}$ such as Bernoulli measurements or Fourier measurements. In this paper, we will prove that even when a subgaussian vector $ \xi^{\left(i\right)} \in \mathbb{C}^m $ does not fulfill a small-ball probability assumption, the \textit{PhaseLift} method is still able to reconstruct a large class of signals $x_0 \in \mathbb{R}^n$ from the measurements. This extends recent work by Krahmer and Liu from the real-valued to the complex-valued case. However, our proof strategy is quite different and we expect some of the new proof ideas to be useful in several other measurement scenarios as well. We then extend our results $x_0 \in \mathbb{C}^n $ up to an additional assumption which, as we show, is necessary.
\end{abstract}

\section{Introduction}\label{section:introduction}
Phase retrieval refers to the problem of reconstructing an unknown vector $x_0 \in \mathbb{C}^n$ from $m$ measurements of the form
\begin{equation}\label{equ:measurements}
y_i = \vert \langle \xi^{\left(i\right)}, x_0 \rangle \vert^2 + w_i \quad \quad \left(\text{where } i \in \left[m\right] \right),
\end{equation}
where the $\xi^{\left(i\right)} \in \mathbb{C}^n$ are known measurement vectors and $w_i \in \mathbb{R}$ represents additive noise. Such problems are ubiquituous in many areas of science and engineering such as X-ray crystallography \cite{harrison1993phase,millane1990phase}, astronomical imaging \cite{fienup1987phase}, ptychography \cite{rodenburg2008ptychography}, and quantum tomography \cite{kueng2017low}.\\
The foundational papers \cite{candes2013phaselift,demanet2014stable,candes2014solving} proposed to reconstruct $x_0$ via the PhaseLift method, a convex relaxation of the original problem.
These papers have triggered many follow-up works since they were the first to establish rigorous recovery guarantees under the assumption that the measurement vectors $\xi^{\left(i\right)}$ are sampled uniformly at random from the sphere. Since then several papers have analyzed scenarios where the measurement vectors possess a significantly reduced amount of randomness, in particular spherical designs \cite{gross2015partial} and coded diffraction patterns \cite{candes2015phase,gross2017improved}. However, the theoretical results for coded diffraction patterns rely on the assumption that the modulus of the illumination patterns is varying. Indeed, it was shown in \cite{fannjiang1} that for certain illumination patterns with constant modulus ambiguities can arise, i.e., it is not possible to determine $x_0$ uniquely from the measurements $y_i$. In fact, such ambiguities can already arise in much simpler settings, where the measurement vectors $ \left( \xi^{\left(i\right)} \right)$ are i.i.d. subgaussian. For example, consider the case that $\xi^{\left(i\right)} = \left(  \varepsilon^{\left(i\right)}_{1} , \ldots, \varepsilon^{\left(i\right)}_{n} \right) $, where the $  \varepsilon^{\left(i\right)}_{j} $ are i.i.d. Rademacher random variables. That is, they only take the values $+1$ and $-1$ each with probability $\frac{1}{2}$. In this case the vector $x_0:=e_1=\left(1, 0, \ldots, 0\right)$ can never be distinguished from the vector $ \tilde{x}_0:=e_2=\left(0, 1, \ldots, 0\right)$. Note that in this scenario it holds that $\mathbb{E} \left[ \big\vert \xi^{\left(i\right)}_j \big\vert^4 \right]  = \mathbb{E} \left[ \big\vert \xi^{\left(i\right)}_j \big\vert^2 \right] $ and, hence, the vector $ \xi^{\left(i\right)} $ does not fulfill a small-ball probability assumption, which means that there is no constant $c>0$ such that for all $\varepsilon>0$ and for all vectors $x  $ it holds that
\begin{equation}
\mathbb{P} \left( \vert \langle x, \xi^{\left(i\right)}  \rangle \vert \le \varepsilon \Vert x \Vert  \right) \le c\varepsilon.
\end{equation}
When the signals are complex even additional classes of ambiguities can arise. For example, when the measurement vectors $ \xi^{\left(i\right)} $ are real, any signal $x$ and its complex-conjugate signal $ \overline{x} $ will result in identical observations.\\

For these reasons, previous works on phase-retrieval from subgaussian measurements (see, e.g., \cite{goldsmith1}) work with real signals and require that all entries of the vector $  \xi^{\left(i\right)} $ fulfill
\begin{equation}\label{ineq:nosmallball}
\mathbb{E} \left[ \big\vert \xi^{\left(i\right)}_j \big\vert^4 \right]  > \mathbb{E} \left[ \big\vert \xi^{\left(i\right)}_j \big\vert^2 \right]
\end{equation}
for all $ j \in \left[n\right] $ or make even stronger assumptions.\\
The only exception is \cite{krahmer2018phase} which shows for the real-valued case ($x_0 \in \mathbb{R}^n$ and $\xi^{\left(i\right)} \in \mathbb{R}^n $) PhaseLift recovers a large class of signals from subgaussian measurements even if estimates of the type \eqref{ineq:nosmallball} are not satisfied. More precisely, one obtains that all signals $x_0$ whose peak-to-average power ratio satisfies a mild bound of the form
\begin{equation} 
\frac{\Vert x_0 \Vert_{\infty}}{\Vert x_0 \Vert} \le \mu < 1
\end{equation}
for some absolute constant $\mu>0$, can be recovered with high probability as long as $ m \gtrsim n $.  However, as the approach in \cite{krahmer2018phase} is intrinsically based on arguments in \cite{shahar_eldar} it cannot be generalized to the complex case in a straightforward manner. This paper provides an analysis both for real-valued and complex-valued signals. We believe that this understanding will be of importance for the subsequent study of structured scenarios such as coded diffraction patterns, which are also intrinsically complex in nature.\\

While the proofs in previous papers \cite{candes2015phase,gross2015partial,gross2017improved,krahmer2018phase} relied on the construction of a so-called dual certificate, our paper will employ a more geometric approach based on Mendelson's small ball method \cite{koltchinskii2015bounding,mendelson2014learning}. This is motivated by recent work \cite{kueng2017low,kabanava2016stable,krahmer2019convex}, which showed that a geometric analysis based on the descent cone of the trace norm can often yield additional insights compared to an approach based on dual certificates. 

For the problem studied in this paper, however, the small-ball method cannot be applied directly to the entire descent cone or the entire  cone of directions in which positive semidefiniteness is preserved. Rather we divide the latter cone into two parts: One that contains all the problematic cases, but is small, and one that is larger, but easier to analyze. Then we control one of these cones using a restricted isometry property and one via the small-ball method.

We think that this novel viewpoint and also some of the techniques developed in this paper will be useful for the analysis of other interesting measurement scenarios, such as the case of heavy-tailed measurement vectors $\xi^{\left(i\right)} $ or the case that $\xi^{\left(i\right)}$ has only entries $0$ and $1$.\\

\section{Background and main results}

\subsection{Notation}

$\mathcal{S}^n$ denotes the vector space of all Hermitian matrices in $\mathbb{C}^{n \times n} $. By $\mathcal{S}^n_{+} \subset \mathcal{S}^n $ we will denote the set of all positive definite Hermitian matrices. For $A,B \in \mathcal{S}^n$ the Hilbert-Schmidt inner product is defined by $\innerproduct{A,B}:= \trace \left(A^* B\right) $. The corresponding norm will be denote by $ \Vert \cdot \Vert_{HS}$. For a matrix $Z \in \mathcal{S}^n $ we will denote their eigenvalues by $\lambda_1 \left(Z\right), \lambda_2 \left(Z\right), \ldots, \lambda_n \left(Z\right) $, which are assumed to be arranged in decreasing order, i.e., $ \lambda_1 \left(Z\right) \ge \lambda_2 \left(Z\right) \ge \ldots \ge \lambda_n \left(Z\right) $. If no confusion can arise, we will suppress the dependence on $Z$ and write $\lambda_i  $ instead of $\lambda_i  \left(Z\right) $. By $\Vert Z \Vert_1$ we will denote the Schatten-$1$ norm of $Z$, i.e. $ \Vert Z \Vert_1 := \sum_{i=1}^{n} \vert \lambda_i \left(Z\right) \vert $. By $\diag \left(Z\right) \in \mathcal{S}^n$ we denote the matrix, which we obtain by setting all off-diagonal entries of $Z$ equal to zero. We will write $a \lesssim b$ or $b \gtrsim a $ if there is a universal constant $C>0$ such that $ a \le Cb $.
\subsection{PhaseLift}
The PhaseLift method was first introduced in \cite{candes2013phaselift}. In this paper we focus on a variant \cite{candes2014solving,demanet2014stable} based on the observation that the measurements $y_i$ can be rewritten in the form
\begin{equation}
y_i = \trace \left(\xi^{\left(i\right)} \left( \xi^{\left(i\right)} \right)^*  X_0  \right)+w_i,
\end{equation}
where $X_0 = x_0 x^*_0 $ is a rank-$1$ matrix encoding the signal to be recovered up to the true inherent phase ambiguity.
From this observation, PhaseLift relaxes the constraint that $X_0$ is of rank $1$ to obtain the optimization problem
\begin{equation}\label{opt:intern1}
\begin{split}
\text{minimize } \quad & \sum_{i=1}^{m} \big\vert \trace \left( \xi^{\left(i\right)} (\xi^{\left(i\right)})^* X  \right) -y_i \big\vert\\
\text{ such that} \quad & X \in \mathcal{S}^n_{+}.
\end{split}
\end{equation}
In order to simplify notation we introduce the linear operator $\mathcal{A}: \mathcal{S}^n \rightarrow \mathbb{R}^m $ as
\begin{equation}\label{definition:Aop}
\mathcal{A} \left(Z\right) \left(i\right) := \innerproduct{ \xi^{\left(i\right)} (\xi^{\left(i\right)})^*, Z }.
\end{equation}
Hence, setting $y:= \left( y_1, \ldots, y_m \right) \in \mathbb{R}^m $, \eqref{opt:intern1} can be rewritten as 
\begin{equation}\label{opt:SDP}
\begin{split}
\text{minimize } \quad & \Vert \mathcal{A} \left(X\right) -y \Vert_{\ell_1}\\
\text{such that} \quad & X \in \mathcal{S}^n_{+}.
\end{split}
\end{equation}
We note that while understanding the relaxation \eqref{opt:SDP} is an important benchmark approach and can be solved in polynomial time, it is typically not practical for applications, as lifting increases the number of optimization variables. For this reason, a very active line of research study recovery guarantees for algorithms that operate in the natural parameter domain such as alternating minimization (see, e.g., \cite{netrapalli2013phase,waldspurger2018alternating}), gradient-descent based formulations (see, e.g., \cite{candes_wirtinger,chen_wirtinger,soltanolkotabi2019structured,wright_phaseretrieval,chen2018gradient}), and \textit{anchored regression} \cite{bahmani2017phase,bahmani2017relax,bahmani2017anchored,goldstein2018phasemax}. 
However, most of these guarantees have been shown under the assumption that the measurement vectors $\left\{ \xi^{\left(i\right)} \right\}_{i=1}^m $ are sampled i.i.d. from the unit sphere, so it will be a natural follow-up of this work to study to which extent our results generalize to the more practical nonconvex algorithms.
In particular, most reconstruction guarantees for these non-convex approaches require an appropriate initialization. For this reason, one needs to study which initialization schemes work for the measurements considered in this paper. A natural approach will be to try spectral initializations and recent generalizations that have been shown to be feasible for a basically minimal number of measurements \cite{lu2017phase,mondelli2020fundamental,dudeja2020analysis,luo2019optimal}. We expect that the analysis provided in this paper will prove useful for this endeavour as the spectral initialization is somewhat connected to trace-norm minimization.

\subsection{Subgaussian measurements}\label{section:subgaussian}
We consider random measurement vectors $ \left\{ \xi^{\left(i\right)} \right\}^m_{i=1} $ given as independent copies of a random vector $\xi $, whose entries $\xi_j$ are assumed to be i.i.d. subgaussian random variables with parameter $K$, expectation $\E \left[\xi_j\right] =0 $, and variance $ \E \left[ \vert \xi_j \vert^2 \right]  =1$. Recall that a random variable $X$ is subgaussian with parameter $K$, if and only if 
\begin{equation}
\inf \left\{  t>0: \  \mathbb{E} \left[ \exp \left( X^2 / t^2 \right) \right]  \le 2  \right\} \le K < + \infty.
\end{equation}
It is well known (see, e.g., \cite{vershynin2016high})  that from this definition it follows for any (measurable) random variable $X$ that
\begin{equation}\label{ineq:subgaussianaux1}
\Vert X \Vert_{L_p} \lesssim \sqrt{p} K.
\end{equation}
Since $ \Vert \xi_1  \Vert_{L_2}^2   = \mathbb{E} \left[ \vert \xi_1 \vert^2 \right] =1 $ inequality \eqref{ineq:subgaussianaux1} immediately implies that
\begin{equation}\label{lemma:Kbound}
K \gtrsim 1.
\end{equation}
Moreover, it is well known (see, e.g., \cite{vershynin2016high}) that for all $ x\in \mathbb{C}^n$ the random variable $ \langle x,\xi \rangle $ is subgaussian with parameter $K \Vert x \Vert $.

\subsection{Previous work}
A number of previous works have studied phase retrieval with subgaussian measurements in the real-valued setting , i.e., $x_0 \in \mathbb{R}^n$ and $\xi \in \mathbb{R}^n$. For measurements fulfilling  $ \mathbb{E} \left[ \big\vert \xi_j \big\vert^4 \right]  > \mathbb{E} \left[ \big\vert \xi_j \big\vert^2 \right] $, \cite{goldsmith1} showed that PhaseLift admits order optimal uniform recovery guarantees.\footnote{In \cite{goldsmith1} instead of \eqref{opt:SDP} the original PhaseLift approach as in \cite{candes2013phaselift} is analysed.} Without the assumption $ \mathbb{E} \left[ \big\vert \xi_j \big\vert^4 \right]  > \mathbb{E} \left[ \big\vert \xi_j \big\vert^2 \right] $, in \cite{krahmer2018phase} the following result was proven, again for the real-valued case.
\begin{theorem}\label{thm:felixykai}\cite[Theorem V.1]{krahmer2018phase}
Let $\xi = \left( \xi_1, \ldots, \xi_n \right) \in \mathbb{R}^n $ be a random vector with i.i.d. subgaussian entries. Then there exist constants $C_1$, $C_2$, $C_3$, and $0 < \mu < 1 $, which depend only on the distribution of $\xi_1$, such that whenever
\begin{equation}
m \ge C_1 n,
\end{equation}
the following statement holds with probability at least $1-\exp \left(- C_2 n \right)$: For all signals $x_0 \in \mathbb{R}^n$ with $\Vert x_0 \Vert_{\infty} \le \mu \Vert x_0 \Vert $ and all noise vectors $ w \in \mathbb{R}^m $ any minimizer of \eqref{opt:SDP} fulfills
\begin{equation}
\Vert \hat{X} - x_0 x^*_0 \Vert_{HS} \le C_3 \frac{\Vert w \Vert_{\ell_1}}{m}.
\end{equation}
\end{theorem}

\section{Main results}
\subsection{Complex signals and complex measurement vectors}
In Theorem \ref{thm:felixykai} both the signal $x_0$ and the measurement vectors $\xi^{\left(i\right)}$ are assumed to be real. While for the measurement vectors this is often too restrictive, the signal $x_0$ is indeed typically real-valued in applications. This important special case will be discussed in Section \ref{section:realsignalcompmeas} below. Nevertheless, we find it still interesting 
from a mathematical point of view under which assumptions recovery is possible for complex-valued signals. Our first result deals with this case.\\

As we have explained in Section \ref{section:introduction}, there are subgaussian distributions for which we cannot achieve uniform recovery of all signals $ x_0 \in \mathbb{C}^n$. For this reason, we define for all $0 < \mu \le 1 $ the set of all signals of mildly bounded peak-to-average power ratio
\begin{equation}\label{def:xmu}
\xmu:= \left\{ x_0 \in \mathbb{C}^n \setminus \left\{ 0\right\} :  \Vert x_0 \Vert_{\infty} \le \mu  \Vert x_0 \Vert    \right\}.
\end{equation}
Indeed, this restriction is very mild as $\mu$ will not depend on the dimension, whereas for a Gaussian random signal the ratio $\frac{\Vert x \Vert_{\infty}}{\Vert x \Vert} $ would scale like $ \sqrt{\frac{\log n}{n}}  $.
Now we are prepared to state the following theorem, which is our first main result.
\begin{theorem}\label{theorem:mainresult}
	Let the observation vector $y$ be given as in \eqref{equ:measurements}, where the random measurement vectors $ \left\{ \xi^{\left(i\right)} \right\}_{i=1}^m $ are defined as in Section \ref{section:subgaussian}. Assume that $ \vert \mathbb{E} \left[  \xi_1^2  \right] \vert^2 \le  1-\beta  $ for some $\beta \in \left(0,1\right) $ and that
	\begin{equation}\label{ineq:condition44}
	m \ge C_1 \frac {K^{20}}{\beta^{5/2}} n.
	\end{equation}
	Then for some probability parameter $p_{\beta}=1- \mathcal{O} \left( \exp \left( \frac{-m \beta^4}{C_2 K^{16}} \right) \right)$ the following two statements hold.
	\begin{enumerate}
		\item With probability at least $p_{\beta}$ one has that  for all vectors $x_0 \in \mathcal{X}_{1/81} $ and any noise vector $w\in \mathbb{R}^m$ any solution $\hat{X}$ of \eqref{opt:SDP} satisfies
		\begin{equation}\label{ineq:bound1}
		\Vert \hat{X} - x_0 x^*_0 \Vert_{1} \le C_3 \frac{K^8}{m \beta^{5/2}} \Vert w \Vert_{\ell_1}.
		\end{equation}
		\item If, in addition, $\mathbb{E} \left[ \vert \xi_1 \vert^4 \right] \ge 1+\beta $, then with probability at least $p_{\beta}$ inequality \eqref{ineq:bound1} holds for all $x_0 \in \mathbb{C}^n \setminus \left\{ 0 \right\} $.
	\end{enumerate}
	Here $C_1$, $C_2$, and $C_3$ are universal constants.
\end{theorem}
The first case of Theorem \ref{theorem:mainresult}, where one makes no assumption on the fourth moment of $\xi_1 $, can be applied also to certain scenarios, where  unique recovery is not possible without this assumption. One important example is that the entries $\xi_i$ are drawn from $ \left\{ z\in \mathbb{C}: \ \Vert z \Vert =1  \right\} $ uniformly at random. Note that these measurements will always yield the same observations $y$ for the two signals
\begin{align}
x_1&=\left(1, 0, \ldots, 0\right),\\
\tilde{x}_1&=\left(0, 1, \ldots, 0\right).
\end{align}
Such very sparse signals are exactly prevented by Condition 1, so there is no contradiction to the theorem's conclusion that unique recovery can be achieved via \eqref{opt:SDP} for all signals $x_0$ such that $\Vert x_0 \Vert_{\infty} \le \frac{1}{82} \Vert x_0 \Vert $.\\

Note that in the second scenario, where assumptions on the fourth moment of $\xi_1$ are available, we obtain a uniform recovery result over all $x_0 \in \mathbb{C}^n $. In the real-valued case a similar result has been shown in \cite{goldsmith1}.

\begin{remark}\label{remark1}
An assumption of the form $ \vert \mathbb{E} \left[  \xi_1^2  \right] \vert^2 \le  1-\beta $ cannot be avoided as the following argument shows. Indeed, if $ \vert \mathbb{E} \left[  \xi_1^2  \right] \vert^2 =1$ the assumption $ \mathbb{E} \left[ \vert   \xi_1^2 \vert  \right] =1  $ implies that $ \xi= \lambda \tilde{\xi} $ almost surely, where $ \lambda \in \left\{ z\in \mathbb{C}: \Vert z \Vert=1 \right\}$ is fixed and $\tilde{\xi} \in \mathbb{R} $ is a real random variable. We observe that
\begin{align}
\big\vert \langle \xi, x_0  \rangle \big\vert^2 &= \big\vert \langle \lambda \tilde{ \xi} , x_0 \rangle \big\vert^2\\
&= \big\vert \langle  \tilde{ \xi} , x_0 \rangle \big\vert^2\\
&= \big\vert \langle  \tilde{ \xi} , \overline{x_0} \rangle \big\vert^2\\
&=\big\vert \langle \xi, \overline{x_0}  \rangle \big\vert^2.
\end{align}
Consequently, $x_0$ and its complex-conjugate $\overline{x_0}$ will always lead to the same measurements.
\end{remark}

\subsection{Real signals and complex measurement vectors}\label{section:realsignalcompmeas}
We have seen in Remark \ref{remark1} that the assumption  $ \vert \mathbb{E} \left[  \xi_1^2  \right] \vert^2 \le  1-\beta $ is necessary to distinguish between a signal $x_0$ and $\overline{x_0}$. However, if, as in many practical applications, it is known a priori that the signal $x_0$ is real-valued then this ambiguity cannot arise and we can uniquely recover without additional assumptions via the following natural variant of the PhaseLift method, where we restrict the search space to real-valued matrices.
\begin{equation}\label{opt:SDP_2}
\begin{split}
\text{minimize } \quad & \Vert \mathcal{A} \left(X\right) -y \Vert_{\ell_1}\\
\text{such that} \quad & X \in \mathcal{S}^n_{+} \cap \mathbb{R}^{n \times n}.
\end{split}
\end{equation}
The following theorem shows that in this scenario the assumption $\vert \mathbb{E} \left[  \xi_1^2  \right] \vert^2 \le  1-\beta $ is indeed not necessary.

\begin{theorem}\label{theorem:mainresult2}
	Let the observation vector $y$ be given as in \eqref{equ:measurements}, where the random measurement vectors $ \left\{ \xi^{\left(i\right)} \right\}_{i=1}^m $ are as defined in Section \ref{section:subgaussian}. 
	Then  the following two statements hold. 
	
	\begin{enumerate}
	\item Assume that
	\begin{equation}\label{ineq:condition44_2}
	m \ge C_1 {K^{20}} n.
	\end{equation}
	Then, with probability at least $1- \mathcal{O} \left( \exp \left( \frac{-m }{C_2 K^{16}} \right) \right)$ one has that for all vectors $x_0 \in \mathcal{X}_{1/81}\cap \mathbb{R}^n$ and any noise vector $w\in \mathbb{R}^m$ any solution $\hat{X}$ of \eqref{opt:SDP_2} satisfies
	\begin{equation}\label{ineq:bound1_2}
	\Vert \hat{X} - x_0 x^*_0 \Vert_{1} \le C_3 {K^8}m  \Vert w \Vert_{\ell_1}.
	\end{equation}
	
	\item If, in addition, it holds that $\mathbb{E} \left[ \vert \xi_1 \vert^4 \right] \ge 1+\beta $ for some  $\beta \in (0,1] $, then, under the refined assumption
	\begin{equation}\label{ineq:condition44_2}
		m \ge C_1 \frac {K^{20}}{\beta^{5/2}} n, 
	\end{equation}
	 one has a more general bound. Namely it holds that with probability at least $1- \mathcal{O} \left( \exp \left( \frac{-m \beta^4}{C_2 K^{16}} \right) \right)$  and for all vectors $x_0 \in \mathbb{R}^n \backslash \left\{0\right\}$, again for arbitrary noise vectors $w\in \mathbb{R}^m$, any solution $\hat{X}$ of \eqref{opt:SDP_2} satisfies
		\begin{equation}\label{ineq:bound1_2}
			\Vert \hat{X} - x_0 x^*_0 \Vert_{1} \le C_3 \frac{K^8}{m \beta^{5/2}} \Vert w \Vert_{\ell_1}.
		\end{equation}
	\end{enumerate}	
	Here $C_1$, $C_2$, and $C_3$ are universal constants.
\end{theorem}


\begin{remark}
In comparison to Theorem \ref{thm:felixykai} the probability bound in Theorem \ref{theorem:mainresult} and Theorem \ref{theorem:mainresult2} is slightly better, as it improves from $1-\exp \left(- \Omega \left( n\right) \right) $ to $ 1-\exp \left(- \Omega \left( m\right) \right)  $.
Moreover, note that in contrast to Theorem \ref{thm:felixykai} the dependence on the subgaussian distribution of $\xi$ is not hidden in the constants. Also note that in our result the dependence on $\beta$ is stated explicitly. However, we do not know whether these bounds are optimal with respect to $K$ and $\beta$.
\end{remark}

\section{Proof of main results}

\subsection{Proof of Theorem \ref{theorem:mainresult}}\label{subsection:proof1}

Our goal is to show that with high probability the matrix $x_0 x_0^*$ is close to the minimizer $\hat{X}$ of the expression $\Vert \mathcal{A}(W)-y\Vert_{\ell_1}$ over all $W\in \mathcal{S}^n_{+}$. A common proof strategy that we will also follow is to establish that all $X \in \mathcal{S}^n_{+}$ with 
\begin{equation}\label{equation:suffcondition1}
\Vert \mathcal{A}(X)-y\Vert_{\ell_1} \le \Vert \mathcal{A}(x_0 x_0^*)-y\Vert_{\ell_1} = \Vert w\Vert_{\ell_1}
\end{equation}
are sufficiently close to the true solution in $\Vert \cdot \Vert_{1} $-norm. More precisely, a sufficient condition for inequality \eqref{ineq:bound1} is that every $X$ fulfilling condition \eqref{equation:suffcondition1} satisfies
\begin{equation}\label{ineq:bound2}
\Vert X - x_0 x^*_0 \Vert_{1} \le C_3 \frac{K^8}{m \beta^{5/2}} \Vert w \Vert_{\ell_1}.
\end{equation} 
Setting $Z= X- x_0 x_0^*$,  equation \eqref{equation:suffcondition1} reads
\begin{equation}\label{equation:suffcondition2}
 \Vert \mathcal{A}(Z)-w\Vert_{\ell_1} \le  \Vert w\Vert_{\ell_1}.
\end{equation}
By the triangle inequality this implies that
\begin{equation}
\Vert \mathcal{A}(Z)\Vert_{\ell_1} < 2 \Vert w \Vert_{\ell_1}.
\end{equation} 
Hence, the upper bound \eqref{ineq:bound2} that we aim to establish directly follows from an appropriate lower bound for $\Vert \mathcal{A} \left(Z\right) \Vert_{\ell_1}/ \Vert Z \Vert_1 $.
Here $Z \in \mathcal{S}^n$ ranges over those matrices for which $x_0 x^*_0 + Z  $ is positive semidefinite. This set is convex, so it is locally well-approximated by a convex cone. To establish a uniform recovery result over all $x_0 \in \xmu $, we need to study the union of the corresponding cones as given by
\begin{equation}
\mmu := \text{cone} \left\{ Z \in \mathcal{S}^n  : \exists x_0\in \xmu \text{ such that } x_0x^*_0 +  Z\in \posdef   \right\}.
\end{equation}
We will refer to this set as the \emph{cone of admissible directions}.

With this notation, our proof strategy can be summarized as establishing a lower bound for 
\begin{equation}
\lambda_{\min} \left( \mathcal{A}, \mathcal{M}_{\mu} \right) := \underset{Z \in  \mathcal{M}_{\mu} \setminus \left\{ 0\right\}  }{\inf} \frac{\Vert \mathcal{A} \left(Z\right) \Vert_{\ell_1} }{\Vert Z \Vert_1},
\end{equation}
which in the literature is commonly referred to as the minimum conic singular value (see, e.g., \cite{tropp2015convex,krahmer2019convex}).
Except for the precise nature of the cone under consideration, this strategy is exactly analogous to a number of works in the recent literature on linear inverse problems \cite{chandrasekaran2012convex,kueng2017low}. In particular, the following lemma, which summarizes our motivating considerations above, can be seen as a variant of 
 \cite[Proposition 2.2]{chandrasekaran2012convex}.
\begin{lemma}\label{lemma:chandrasekaranvariant}
Let $\mathcal{A}$ be the operator defined in (\ref{definition:Aop}). Assume that $y= \mathcal{A} \left( x_0 x^*_0 \right) +w $. Then the minimizer $\hat{X}$ of \eqref{opt:SDP} satisfies
\begin{equation}
\Vert \hat{X}-x_0 x^*_0 \Vert_{1} \le \frac{2 \Vert w \Vert_{\ell_1}}{\lambda_{\min} \left( \mathcal{A}, \mathcal{M}_{\mu} \right) }.
\end{equation}
\end{lemma}
In the following, our goal will be to derive an appropriate lower bound for $\lambda_{\min} \left( \mathcal{A}, \mathcal{M}_{\mu} \right)$. One difficulty in the analysis is that not all matrices belonging to $\mathcal{M}_{\mu} $ are positive semidefinite. Indeed, in this scenario one could use that for positive semidefinite matrices an approximate $\ell_1$ isometry holds (see, e.g. \cite[Section 3]{candes2013phaselift}). While not all matrices in $\mathcal{M}_{\mu}$ are positive semidefinite the following lemma states that each matrix belonging to $\mathcal{M}_{\mu} $ possesses at most one negative eigenvector.
\begin{lemma}\label{lemma:onenegativeeigen}
Suppose that $Z\in \mmu$. Then $Z$ has at most one strictly negative eigenvalue.
\end{lemma}
\begin{proof}
Let $Z\in \mmu$. By definition of $ \mmu $ we can find $x_0 \in \xmu$ and $t>0 $ such that 
\begin{equation}\label{ineq:intern1}
x_0x^*_0 + tZ \in \mathcal{S}^n_{+}.
\end{equation}
Suppose now by contradiction that $Z$ has two (strictly) negative eigenvalues with corresponding eigenvectors $ z_1, z_2 \in \mathbb{C}^n$.  Then we can find a vector $u \in \text{span} \left\{ z_1, z_2 \right\} \backslash \left\{ 0 \right\}  $ such that $ \langle  u,x_0 \rangle =0 $. This implies that for any $ t >0 $ we have that
\begin{equation}
u \left(  x_0 x^*_0 + t Z \right) u^* = t u^*Zu < 0,
\end{equation} 
which is a contradiction to (\ref{ineq:intern1}).
\end{proof}

Recall that for a matrix $Z \in \mathcal{S}^n$ we denoted its eigenvalues by $\left\{ \lambda_i \left(Z\right) \right\}^n_{i=1} $ in decreasing order. By the previous lemma it holds that $ \lambda_i \left(Z\right)  \ge 0$ for all $ i \in \left[n-1\right] $ and all $Z \in \mathcal{M}_{\mu} $. For the proof we will partition $\mmu$ into two sets. Namely, for $ \alpha > 0 $ we define
\begin{align}
\mathcal{M}_{1,\mu, \alpha} &:= \left\{Z\in \mmu:  -\lambda_n \left(Z\right) \le \alpha  \sum_{i=1}^{n-1} \lambda_i \left( Z \right)  \right\},\\
\mathcal{M}_{2,\mu, \alpha}& := \left\{Z\in  \mmu:   -\lambda_n \left(Z\right) > \alpha  \sum_{i=1}^{n-1} \lambda_i \left( Z \right)   \right\}.
\end{align}
The two sets can be interpreted in the following way. If we would suppose that $\alpha=1 $ it would follow that $\trace \left(Z\right) < 0 $ for all matrices $Z\in \mualphatwo$. In particular, this implies that there is $x_0 \in \xmu$ such that $Z$ is in the descent cone of the function $\trace \left( \cdot \right)$ at the point $x_0 x^*_0 $. Hence, for $ \alpha < 1 $ we can interpret $\mualphatwo $ as a slightly enlarged union of descent cones. In order to bound $ \underset{Z \in \mualphatwo}{\inf} \Vert \mathcal{A} \left(Z\right) \Vert_{\ell_1}  / \Vert Z \Vert_1 $ from below we will rely on the following lemma, which is proven in Section \ref{section:smallballapplied}.
\begin{lemma}\label{lemma:smallballapplied}
	Assume that one of following two conditions is satisfied for $ \beta \in (0,1] $:
	\begin{enumerate}
		\item It holds that $ \vert \mathbb{E} \left[  \xi_1^2  \right] \vert^2 \le  1-\beta  $. In this case we set $\mu=1/81$.
		\item In addition to $ \big\vert \mathbb{E} \left[  \xi_1^2  \right] \big\vert^2 \le  1-\beta $, the inequality  $\mathbb{E} \left[ \vert \xi_1 \vert^4 \right] \ge 1+\beta $, is fulfilled. In this case we set $ \mu =1 $.
	\end{enumerate}
	Moreover, assume that
	\begin{equation}\label{ineq:condition2}
	m \ge C_1 \frac {K^{20}}{\beta^{5}} n.
	\end{equation}
	Then with probability at least $1- 2\exp \left( \frac{-m\beta^4}{C_2 K^{16}}   \right) $ it holds that
	\begin{equation}\label{ineq:key22}
	\underset{Z\in \mualphatwo \setminus \left\{ 0 \right\}}{\inf} \frac{\Vert \mathcal{A} \left(Z\right) \Vert_{\ell_1}}{ \Vert Z\Vert_{1}} \ge C_3 \frac{\beta^{5/2}}{K^8} m ,
	\end{equation}
	where $ \alpha = 4/5 $. Here $C_1$, $C_2$, and $C_3$ are universal constants.
\end{lemma}
The proof of Lemma \ref{lemma:smallballapplied} makes use of the fact that the set $ \mualphatwo $ has low complexity in the sense that the matrices in $ \mualphatwo $ are approximately low-rank.\\

In contrast, the set $ \mualphaone$ has rather high complexity. For example, note that $\mathcal{S}^n_{+} \subset \mualphaone $. Nevertheless, the quantity $	\underset{Z\in \mualphaone \setminus \left\{ 0 \right\}}{\inf} \Vert \mathcal{A} \left(Z\right) \Vert_{\ell_1} / \Vert Z\Vert_{1} $ can be bounded from below, because the measurement matrices $\xi^{\left(i\right)} (\xi^{\left(i\right)})^* $  are positive semidefinite and the matrices in $\mualphaone$ also have a dominant positive semidefinite component. This is achieved by the following lemma, whose proof can be found in Section \ref{section:oneestimate}.
\begin{lemma}\label{lemma:oneestimate}
	Let $0 < \mu \le 1$, $ \alpha >0$, and $ \delta>0$. Assume that
	\begin{equation}
	m \ge  \frac{C_1K^4n}{\delta^2}.
	\end{equation}
	Then with probability at least $1- \mathcal{O} \left(  \exp \left( -\frac{m}{C_2K^4}  \right) \right) $ for all $Z\in \mathcal{M}_{1,\mu,\alpha} $ it holds that
	\begin{equation}\label{ineq:intern4}
	\frac{1}{m} \Vert \mathcal{A} \left(Z\right) \Vert_{\ell_1} \ge \frac{1-\delta-\alpha-\alpha\delta}{1+ \alpha} \Vert Z \Vert_1.
	\end{equation}
	Here $C_1$ and $C_2$ are absolute constants.
\end{lemma}
We remark that Lemma \ref{lemma:oneestimate} would no longer hold if the measurement matrices $\xi^{\left(i\right)} (\xi^{\left(i\right)})^* $ would be replaced by  symmetric matrices with i.i.d. Gaussian entries (see \cite[Proposition 1]{slawski2015regularization}).\\

Having gathered all the necessary ingredients we can prove the main result  of this manuscript.
\begin{proof}[Proof of Theorem \ref{theorem:mainresult}]
	Set $\alpha =4/5$. The proof of the two statements is analogous, except that for the first statement we set $ \mu =1/81 $ whereas for the second statement we set $ \mu =1 $. By Lemma \ref{lemma:oneestimate} and assumption (\ref{ineq:condition44}) it follows that with probability at least $1-\mathcal{O} \left( \exp \left( -\frac{m}{CK^4} \right) \right) $
	\begin{equation}\label{ineq:intern48}
	\underset{Z\in \mualphaone \setminus \left\{ 0 \right\}	}{\inf} \frac{\Vert \mathcal{A} \left(Z\right) \Vert_{\ell_1}}{ \Vert Z \Vert_1} \gtrsim m.
	\end{equation}
	Furthermore, by Lemma \ref{lemma:smallballapplied} we have with probability at least $1- 2\exp \left( \frac{-m\beta^4}{C_2 K^{16}}   \right) $ that
	\begin{equation}\label{ineq:intern49}
	\underset{Z \in \mualphatwo \setminus \left\{ 0 \right\}}{\inf} \frac{\Vert \mathcal{A} \left(Z\right) \Vert_{\ell_1}}{ \Vert Z \Vert_{1}} \gtrsim \frac{\beta^{5/2}}{K^8}m
	\end{equation}
	holds.\\

	Set $Z:= \hat{X} - x_0 x^*_0$. Note that by definition we have that $Z$ is an admissible direction, i.e., $Z\in \mathcal{M}_{\mu}$. It follows by (\ref{ineq:intern48}), (\ref{ineq:intern49}), and $\mathcal{M}_{\mu}= \mualphaone \cup \mualphatwo $ that
	\begin{equation}
	\lambda_{\min} \left(  \mathcal{A}, \mathcal{M}_{\mu}  \right) \gtrsim \min \left\{ 1; \frac{\beta^{5/2}}{K^8}   \right\} m \ge  \frac{\beta^{5/2}}{K^8}  m,
	\end{equation}
	where in the last inequality we used (\ref{lemma:Kbound}) and $0<\beta  \le 1 $. It follows by Lemma \ref{lemma:chandrasekaranvariant} that
	\begin{equation}
	\Vert \hat{X} - x_0 x^*_0 \Vert_1 \le \frac{2 \Vert w \Vert_{\ell_1}}{ 	\lambda_{\min} \left( \mathcal{M}_{\mu}  \right)  } \lesssim  \frac{K^8}{m \beta^{5/2}}   \Vert w \Vert_{\ell_1},
	\end{equation}
	which finishes the proof.
\end{proof}

\subsection{Proof of Theorem \ref{theorem:mainresult2}}
The proof of Theorem \ref{theorem:mainresult2} is in large parts analogous to the proof of Theorem \ref{theorem:mainresult}. For this reason,  we will only highlight the main differences. Replacing $\xmu $ by $ \xmu \cap \mathbb{R}^n $ and $ \mmu $ by $\mmu \cap  \mathbb{R}^{n \times n} $ we can argue analogously to Section \ref{subsection:proof1} with the only difference that Lemma \ref{lemma:smallballapplied} has to be replaced by the following variant.
\begin{lemma}\label{lemma:smallballapplied_variant}
	Assume that one of following two conditions is satisfied for $ \mu, \beta \in (0,1] $:
	\begin{enumerate}
		\item It holds that $\mu  = \frac{1}{81} $ and $ \beta =1 $.
		\item It holds that $\mathbb{E} \left[ \vert \xi_1 \vert^4 \right] \ge 1+\beta $ and $\mu=1 $.
	\end{enumerate}
	Moreover, assume that
	\begin{equation}
	m \ge C_1 \frac {K^{20}}{\beta^{5}} n.
	\end{equation}
	Then with probability at least $1- 2\exp \left( \frac{-m\beta^4}{C_2 K^{16}}   \right) $ it holds that
	\begin{equation}
	\underset{Z\in (\mualphatwo \cap \mathbb{R}^{n \times n} ) \setminus \left\{ 0 \right\}}{\inf} \frac{\Vert \mathcal{A} \left(Z\right) \Vert_{\ell_1}}{ \Vert Z\Vert_{1}} \ge C_3 \frac{\beta^{5/2}}{K^8} m ,
	\end{equation}
	where $ \alpha = 4/5 $. Here $C_1$, $C_2$, and $C_3$ are universal constants.
\end{lemma}
In order to prove Lemma \ref{lemma:smallballapplied_variant} we can proceed similarly as in the proof of Lemma \ref{lemma:smallballapplied}, see Section \ref{section:smallballapplied}, where in the proof of Lemma \ref{lemma:smallballapplied} we have highlighted the necessary modifications.

\section{Proof of Lemma \ref{lemma:oneestimate}}\label{section:oneestimate}

\begin{proof}
Note that for any $z\in \mathbb{C}^n $ we have that $\Vert \mathcal{A} \left(zz^*\right) \Vert_{\ell_1} = \sum_{i=1}^{m} \vert \langle \xi_i , z \rangle \vert^2 $. Let $A\in \mathbb{C}^{m\times n}$ be the matrix whose rows are given by $ \left\{ \xi_i \right\}_{i=1}^m $. It follows that $ \Vert \mathcal{A} \left(zz^*\right) \Vert_{\ell_1}  = \Vert Az \Vert^2 $. It follows from  \cite[Theorem 4.6.1]{vershynin2016high} that due to our assumption on $m$ with probability at least $ 1- \mathcal{O} \left(  \exp \left( -\frac{m}{CK^4}  \right) \right) $ for all $z\in \mathbb{C}^m $ it holds that
\begin{equation}
\left(1-\delta\right) \Vert z \Vert^2 \le \frac{1}{m} \Vert Az \Vert^2 \le \left(1+\delta\right) \Vert z \Vert^2.
\end{equation}
Due to the observation above this is equivalent to 
\begin{equation}\label{equ:rank1rip}
\left(1-\delta\right) \Vert z \Vert^2 \le \frac{1}{m} \Vert \mathcal{A} \left(zz^*\right) \Vert_1 \le \left(1+\delta\right) \Vert z \Vert^2
\end{equation}	
for all $z \in \mathbb{C}^n $. We will assume in the following that (\ref{equ:rank1rip}) holds for all $z \in \mathbb{C}^m$.\\

Let $Z\in \mathcal{M}_{1,\mu, \alpha}$ with corresponding eigenvalue decomposition $ Z = \sum_{i=1}^{n} \lambda_i v_i v^*_i  $. We observe that
\begin{align}
 \frac{1}{m}  \Vert \mathcal{A} \left(Z\right) \Vert_{\ell_1} &=  \frac{1}{m}  \sum_{j=1}^{m} \vert  \left(\xi^{\left(j\right)}\right)^*  Z \xi^{\left(i\right)} \vert \\
& \ge  \frac{1}{m}  \sum_{j=1}^{m} \left(\xi^{\left(j\right)}\right)^* Z \xi^{\left(j\right)}\\
& =  \frac{1}{m}   \sum_{j=1}^{m} \left(\xi^{\left(j\right)}\right)^*\left( \sum_{i=1}^{n} \lambda_i v_i v^*_i   \right) \xi^{\left(j\right)}\\
&= \frac{1}{m} \sum_{i=1}^{n} \lambda_i \sum_{j=1}^{m} \left(\xi^{\left(j\right)}\right)^* v_i v^*_i \xi^{\left(j\right)}\\
&= \frac{1}{m} \sum_{i=1}^{n} \lambda_i  \Vert  \mathcal{A} \left(v_i v^*_i\right) \Vert_{\ell_1}.
\end{align}
By Lemma \ref{lemma:onenegativeeigen} we know that $Z$ has at most one negative eigenvalue. If all eigenvalues $\lambda_i \left(Z\right) $ are positive, this inequality chain and inequality (\ref{equ:rank1rip}) imply that
\begin{equation}
\frac{1}{m} \Vert \mathcal{A} \left(Z\right) \Vert_{\ell_1} \ge \left(1-\delta\right) \sum_{i=1}^{n} \lambda_i = \left(1-\delta\right) \Vert Z \Vert_1,
\end{equation}
which shows (\ref{ineq:intern4}). Now suppose that $\lambda_n \left(Z\right) <0$. By (\ref{equ:rank1rip}) and $ -\lambda_n \left(Z\right) \le \alpha \sum_{i=1}^{n-1} \lambda_i \left(Z\right) $, which is due to $Z\in \mathcal{M}_{1,\mu,\alpha} $, we obtain that
\begin{equation}\label{ineq:intern2}
\begin{split}
 \frac{1}{m}  \Vert \mathcal{A} \left(Z\right) \Vert_1 &\ge  \left( 1-\delta \right) \sum_{i=1}^{n} \lambda_j  + \left(1+\delta \right) \lambda_n\\
 &\ge  \left(1-\delta - \alpha \left(1+\delta\right)\right)  \sum_{i=1}^{n-1} \lambda_i.
\end{split}
\end{equation}
Again using the relation $ -\lambda_n \left(Z\right) \le \alpha \sum_{i=1}^{n-1} \lambda_i \left(Z\right) $ we can also observe that
\begin{equation}\label{ineq:intern3}
\Vert Z \Vert_1 =  \sum_{j=1}^{n-1} \lambda_j  -\lambda_n \le  \left(1+\alpha \right) \sum_{i=1}^{n-1} \lambda_j.
\end{equation}
Combining (\ref{ineq:intern2}) and (\ref{ineq:intern3}) shows (\ref{ineq:intern4}), which finishes the proof.
\end{proof}

\section{Proof of Lemma \ref{lemma:smallballapplied} and Lemma \ref{lemma:smallballapplied_variant}}\label{section:smallballapplied}

In order to prove Lemma \ref{lemma:smallballapplied} and Lemma \ref{lemma:smallballapplied_variant} we will use the following version of Mendelson's small ball method \cite{koltchinskii2015bounding,mendelson2014learning}, a tool for deriving a lower bound for nonnegative empirical process.

\begin{lemma}\cite[Lemma 1]{dirksen2018gap}\label{thm:smallballmethod}
	Let $\mathcal{Z} \subset \mathcal{S}^n $ and let $\xi^{(1)}, \xi^{(2)}, \ldots, \xi^{(m)} $ be i.i.d. random vectors. Let $u>0$ and $t>0$ and define
	\begin{equation}\label{definition:Q}
	Q_{\mathcal{Z}} \left( u \right) := \underset{Z \in \mathcal{Z}}{\inf} \mathbb{P} \left( \vert \innerproduct{  \xi^{\left(1\right)} \left(\xi^{\left(1\right)}\right)^*, Z } \vert   \ge u \right).
	\end{equation}
	Then, with probability at least $1-2\exp \left( -2t^2 \right) $, it holds that
	\begin{align}
	&\underset{Z \in \mathcal{Z}}{\inf}   \left( \frac{1}{m} \sum_{i=1}^{m} \big\vert \innerproduct{ \xi^{\left(i\right)} \left(\xi^{\left(i\right)}\right)^*, Z} \big\vert \right)\\
	 \ge & u \left(  Q_{\mathcal{Z}} \left( 2u \right) - \frac{4}{u} \mathbb{E} \left[ \underset{Z \in \mathcal{Z}}{\sup} \Big\vert  \frac{1}{m} \sum_{i=1}^{m} \varepsilon_i \innerproduct{ \xi^{\left(i\right)} \left(\xi^{\left(i\right)}\right)^* , Z }   \Big\vert \right]       -\frac{t}{\sqrt{m}}   \right),
	\end{align}
	where $ \left( \varepsilon_{i} \right)^m_{i=1} $ are independent, symmetric, $\left\{ -1,1 \right\} $-valued random variables that are independent of $ \left(X_i \right)^m_{i=1} $.
\end{lemma}
Our goal is to apply Lemma \ref{thm:smallballmethod} to $\mathcal{Z}= \mualphatwo \cap \left\{Z \in \mathcal{S}^n: \ \Vert Z \Vert_F =1 \right\} $. The following key lemma  shows that matrices in $ \mualphatwo $ have two favorable properties: They are approximately low-rank and their mass with respect to the Frobenius norm is not concentrated on the diagonal for $\mu$ is small. The first property follows directly from the fact that the negative eigenvalue is rather small, the second property requires the spectral flatness of $x_0$, i.e., that $\mu$ is bounded.
\begin{lemma}\label{lemma:key1}
Let $ \alpha >0 $ and $ 0 < \mu \le 1 $. Assume that $Z \in \mathcal{M}_{2,\mu, \alpha}$. Then it holds that
\begin{enumerate}
\item 
\begin{equation}\label{ineq:lowrank}
\Vert Z \Vert_1 \le \left(1+  \frac{1}{\alpha} \right) \Vert Z \Vert_{\text{HS}},
\end{equation}
\item 
\begin{align}
\Vert \diagonal{Z} \Vert_{HS}  \le \left(  \sqrt{1 - \frac{1}{ \left( 1 + \alpha^{-1} \right)^2 }}  + 3\mu   \right) \Vert Z \Vert_{HS}.
\end{align}
\end{enumerate}
\end{lemma}

\begin{proof}
Let $Z \in  \mualphatwo$. By definition of $\mualphatwo $ we have that $ \alpha \sum_{i=1}^{n-1} \lambda_i \left(Z\right) < - \lambda_n \left(Z\right)  $, which implies that
\begin{equation}\label{ineq:interna2}
\Vert Z \Vert_1 = \sum_{i=1}^{n-1} \lambda_i \left(Z\right)   -\lambda_n \left(Z\right) \le - \left( 1 + \frac{1}{\alpha}   \right) \lambda_n \left(Z\right) \le    \left( 1 + \frac{1}{\alpha}    \right)  \Vert Z \Vert_{\text{HS}}.
\end{equation}
This proves inequality (\ref{ineq:lowrank}).\\

\noindent In order to prove the second inequality note that by definiton of $\mathcal{M}_{2,\mu,\alpha} \subset \mathcal{M}_{\mu} $ we can choose $x_0 \in \mathcal{X}_{\mu}  \cap S^{n-1} $ such that there exists $t>0$ with $x_0x^*_0 + tZ$ positive semidefinite. For this choice of $x_0$ we can decompose $Z$ uniquely into
\begin{equation}
Z= \underset{=: Z_1}{\underbrace{ -\lambda x_0x^*_0 + ux^*_0+ x_0u^*} }+ Z_2,
\end{equation}
where $ \lambda \in \mathbb{R} $, $ \langle u,x_0 \rangle =0$, and $Z_2x_0=0$. We observe that
\begin{align}\label{ineq:interna5}
\Vert \diagonal{Z} \Vert_{HS} &\le  \Vert \diagonal{Z_1 } \Vert_{HS}  + \Vert \diagonal{Z_2}\Vert_{HS}.
\end{align}
We will bound the two summands separately. We begin with $ \Vert \diagonal{Z_1} \Vert_{HS} $ and observe that
\begin{equation}\label{ineq:interna6}
\begin{split}
\Vert \diagonal{Z_1} \Vert_{\text{HS}} &\le \vert \lambda \vert \Vert \diag \left(x_0 x^*_0\right) \Vert_{HS} + 2 \Vert \diag \left(ux^*_0\right) \Vert_{\text{HS}}\\
&= \vert \lambda \vert \sqrt{\sum_{i=1}^{n} \big\vert \left(x_0\right)_i \big\vert^4  } + 2 \sqrt{\sum_{i=1}^{n} \big\vert \left(x_{0}\right)_i \big\vert^2 \big\vert u_i \big\vert^2  }\\
&\le \mu \vert \lambda \vert + 2\mu \Vert u \Vert \\
&\le 3\mu \Vert Z_1 \Vert_{\text{HS}}\\
&\le 3\mu \Vert Z \Vert_{\text{HS}}.
\end{split}
\end{equation}
In the first inequality we used the triangle inequality and in the third line we used that $\Vert x_0 \Vert_{\infty} \le \mu \Vert x_0 \Vert= \mu  $ due to $x_0 \in \mathcal{X}_{\mu}\cap S^{n-1}$. In the fourth line we used that $\vert \lambda \vert \le \Vert Z_1 \Vert_{HS} $ and $ \Vert u \Vert \le \Vert Z_1 \Vert_{HS} $, which follows from the fact that the summands of $Z_1= -\lambda x_0 x^*_0 + ux^*_0 + x_0 u^* $ are orthogonal to each other. In the last line we again used that $\Vert Z_1 \Vert_{HS} \le \Vert Z \Vert_{HS} $ as $Z$ is decomposed orthogonally into $Z=Z_1+Z_2$. \\

In order to bound $ \Vert \diagonal{Z_2}  \Vert_{HS}$ we note first  that $Z_2$ is positive semidefinite. Indeed, suppose by contradiction that $Z_2$ is not positive semidefinite. Then there would exist a vector $v\in \mathbb{C}^n$ such that $ \langle v, x_0 \rangle =0$ and $ v^* Z_2 v <0$. In particular, this would imply that $v^* \left( x_0 x^*_0 + t Z \right)v <0 $ for all $ t >0$, which is a contradiction to our choice of $x_0$.\\
Now let $w\in \mathbb{C}^n$ be the normalized (i.e., $ \Vert w \Vert=1 $) eigenvector corresponding to the eigenvalue $ \lambda_n \left(Z\right)$. Then we obtain that
\begin{align}\label{ineq:interna3}
\lambda_n \left( Z \right) = w^* Z w = w^* Z_1 w + w^*Z_2w \ge w^* Z_1 w \ge  - \Vert Z_1 \Vert  \ge - \Vert Z_1 \Vert_{HS} ,
\end{align}
where the first inequality follows from the fact that $Z_2$ is positive semidefinite. Using this observation we obtain that
\begin{equation}
\begin{split}
\Vert \diagonal{Z_2}\Vert_{\text{HS}} &\le \Vert Z_2 \Vert_{\text{HS}}\\
&= \sqrt{ \Vert Z \Vert^2_{\text{HS}} - \Vert Z_1 \Vert_{\text{HS}}^2 }\\
& \overset{(\ref{ineq:interna3})}{\le}  \sqrt{ \Vert Z \Vert^2_{\text{HS}} - \lambda^2_n \left(Z\right) } \\
&\le \sqrt{ \Vert Z \Vert^2_{\text{HS}} -   \frac{1}{ \left( 1 + \alpha^{-1}  \right)^2} \Vert Z \Vert_{1}  }\\
& \le \Vert Z \Vert_{\text{HS}} \sqrt{1 - \frac{1}{ \left( 1 + \alpha^{-1}  \right)^2 }},
\end{split}
\end{equation}
where in the fourth line we used that $ - \lambda_n \left(Z\right) \ge      \frac{1}{1+\alpha^{-1}}  \Vert Z \Vert_{1} $, which is a consequence of the first inequality of (\ref{ineq:interna2}). Combining this estimate with (\ref{ineq:interna5}) and (\ref{ineq:interna6}) shows part (2), which finishes the proof.

\end{proof}
In analogy to \cite{koltchinskii2015bounding} we bound $Q_{\mathcal{Z}} \left(2u\right) $ using the following lemma, whose proof is based on the Paley-Zygmund inequality. A key difference is that we use the Hanson-Wright inequality to control the fourth moment $\mathbb{E}  \vert \xi^* A \xi \vert^4  $ appropriately.
\begin{lemma}\label{lemma:PRsmallballprobabilities}
	Let $A \in  \mathcal{S}^n$ and let $ \xi = \left( \xi_1, \ldots, \xi_n \right) $ be a random vector with independent and identically distributed entries $ \xi_i  $ taking values in $ \mathbb{C} $ such that $ \mathbb{E} \xi_i = 0 $, $ \mathbb{E} \vert \xi_i \vert^2 =1 $, and $\Vert \xi_i  \Vert_{\psi_2} \le K $. Then we have that
	\begin{equation}\label{equ:smallballphaseretrieval}
	\mathbb{P} \left(  \vert  \xi^* A \xi \vert^2 \ge \frac{\E \vert \xi^* A \xi \vert^2}{2}  \right) \ge  \frac{ \left( \E \vert \xi^* A \xi \vert^2 \right)^2 }{ C \left( K^8 \Vert A \Vert^4_{HS} +  \left( \trace\left(A\right) \right)^4    \right) }.
	\end{equation}
	Here $C>0 $ is an absolute constant.
\end{lemma}

\begin{proof}
	Note that by the Paley-Zygmund inequality (see, e.g., \cite{decoupling}) we have that for all $ 0 < t \le  \E \vert \xi^* A \xi \vert^2$ 
	\begin{equation}
	\mathbb{P} \left(  \vert  \xi^* A \xi \vert^2 \ge t   \right)  \ge \frac{\left( \E \vert \xi^* A \xi \vert^2  - t  \right)^2}{\mathbb{E}  \vert \xi^* A \xi \vert^4 }.
	\end{equation}
	In particular, setting $t=\E \vert \xi^* A \xi \vert^2/2 $ yields that
	\begin{equation}\label{equ:paleyzygmund_internwef}
	\mathbb{P} \left(  \vert  \xi^* A \xi \vert^2 \ge \frac{\E \vert \xi^* A \xi \vert^2}{2}  \right) \ge \frac{\left( \E \vert \xi^* A \xi \vert^2   \right)^2}{4 \mathbb{E}  \vert \xi^* A \xi \vert^4 }.
	\end{equation}
	To estimate $\mathbb{E}  \vert \xi^* A \xi \vert^4$ from above we note that the triangle inequality yields that
	\begin{equation}\label{ineq:interna1}
	\begin{split}
	\mathbb{E}  \vert \xi^* A \xi \vert^4 &\lesssim  \mathbb{E} \left[ \big\vert \xi^* A \xi  - \E \left[  \xi^* A \xi \right]  \big\vert^4 \right] + \vert \mathbb{E} \xi^* A \xi \vert^4\\
	&=  \mathbb{E} \left[ \big\vert \xi^* A \xi  - \E \left[  \xi^* A \xi \right]  \big\vert^4 \right] +  \left( \trace\left(A\right) \right)^4.
	\end{split}
	\end{equation}
	In order to estimate the first summand we will use that $ \big\vert \xi^* A \xi  - \E \left[  \xi^* A \xi \right]  \big\vert$ has a mixed subgaussian/subexponential tail. 
	We can bound the tail probability using the Hanson-Wright inequality (in the version of \cite{rudelson_hansonwright}), which states that there is a numerical constant $c>0$ such that for all $t>0$ it holds that 
	\begin{equation}
	\mathbb{P} \left(  \vert  \xi^* A \xi - \mathbb{E} \left[ \xi^* A \xi \right] \vert > t   \right) \le 2   \exp \left( -c \ \text{min} \left\{ \frac{t^2}{K^4 \Vert A \Vert^4_{HS}}  , \frac{t}{K^2 \Vert A \Vert}  \right\} \right).
	\end{equation}
	This yields that
	\begin{align}
	\mathbb{E}  \left[ \big\vert \xi^* A \xi - \E \left[ \xi^* A \xi \right]  \big\vert^4 \right] &= 4 \int_{0}^{\infty} t^3 \ \mathbb{P} \left( \big\vert \xi^* A \xi - \E \left[  \xi^* A \xi \right] \big\vert > t  \right)  dt \\
	 &\le 8 \left(     \int_{0}^{\infty} t^3 \  \exp \left( -c  \frac{t^2}{K^4 \Vert A \Vert^2_{HS}} \right) dt     +  \int_{0}^{\infty} t^3 \  \exp \left( -c  \frac{t}{K^2 \Vert A \Vert} \right) dt   \right)\\
	& = 8 \left( K^8 \Vert A \Vert^4_{HS}  \int_{0}^{\infty} u \  \exp \left( -cu^2 \right) du  + K^8 \Vert A \Vert^4    \int_{0}^{\infty} u^3 \  \exp \left( -cu \right) du    \right) \\
	&\lesssim  K^8 \Vert A \Vert^4_{HS},
	\end{align}
	where the third line follows from a change of variables. Combining this inequality chain with (\ref{ineq:interna1}) we obtain that
	\begin{equation}
	\mathbb{E}  \vert \xi^* A \xi \vert^4  \lesssim K^8 \Vert A \Vert^4_{HS} +   \left( \trace\left(A\right) \right)^4  .
	\end{equation}
	Inserting this into (\ref{equ:paleyzygmund_internwef}) finishes the proof. 
\end{proof}
In order to apply Lemma \ref{lemma:PRsmallballprobabilities} we need a lower bound for $\mathbb{E} \left[ \vert \xi^* A \xi \vert^2 \right]$. The next lemma computes this quantity.
\begin{lemma}\label{lemma:PRsecondmoments}
	Let $ \xi = \left( \xi_1, \ldots, \xi_n \right) $ be a random vector with independent and identically distributed entries $ \xi_i  $ taking values in $ \mathbb{C} $ such that $ \mathbb{E} \xi_i = 0 $ and $ \mathbb{E} \vert \xi_i \vert^2 =1 $. Then for all matrices  $ A \in \mathcal{S}^n $ it holds that
	\begin{equation}\label{equ:secondmoment}
	\begin{split}
	&\mathbb{E} \left[ \vert \xi^* A \xi \vert^2 \right]\\
	=& \left( \trace A \right)^2 + \left(  \mathbb{E} \vert \xi_i \vert^4   -1 \right) \sum_{i=1}^{n} A^2_{i,i} + \left( 1 + \big\vert \mathbb{E} \left[ \xi^2_1 \right] \big\vert^2  \right) \sum_{i \ne j} \text{Re} \left( A_{i,j} \right)^2\\
	 &+ \left( 1- \big\vert \mathbb{E}  \left[ \xi^2_1  \right] \big\vert^2  \right) \sum_{i \ne j} \text{Im} \left( A_{i,j} \right)^2.
	\end{split}
	\end{equation}
\end{lemma}

\begin{proof}
	First, we observe that
	\begin{align}
	\mathbb{E} \left[ \vert \xi^* A \xi \vert^2  \right] &= \mathbb{E} \left[   \left( \sum_{i,j}  A_{i,j} \overline{\xi_i} \xi_j \right) \left(  \sum_{i',j'} \overline{ A_{i',j'}} \xi_{i'} \overline{ \xi_{j'}} \right)   \right] \\
	&= \sum_{i,i',j,j'}  \mathbb{E} \left[ \left(  A_{i,j} \overline{\xi_i} \xi_j \right) \left(   \overline{ A_{i',j'}} \xi_{i'} \overline{ \xi_{j'}} \right)    \right]     \\
	& = \sum_{i,j} \mathbb{E} \left[  \left(  A_{i,i} \vert \xi_i \vert^2  \right)  \left(  A_{j,j} \vert \xi_j \vert^2  \right)   \right] +  \sum_{i\ne j ,i' \ne j'} \mathbb{E} \left[   \left(  A_{i,j} \overline{\xi_i} \xi_j \right) \left(  \overline{ A_{i',j'}} \xi_{i'} \overline{ \xi_{j'}} \right)   \right] \\
	& = \left(I\right)+ \left(II\right),
	\end{align}
	where in the third line we used that $ \mathbb{E} \left[ \xi_i \right] = 0 $ and that the entries of $\xi$ are independent, which implies that there are no summands where one index appears exactly three times. The first summand can be computed by
	\begin{align}
	\left(I\right) &= \sum_i A^2_{i,i} \mathbb{E} \left[ \vert \xi_i \vert^4 \right] + \sum_{i \ne j} A_{i,i} A_{j,j} \mathbb{E} \left[ \vert \xi_i \vert^2 \right] \mathbb{E} \left[ \vert \xi_j \vert^2  \right]\\
	&= \sum^n_{i=1} A^2_{i,i} \mathbb{E} \left[ \vert \xi_i \vert^4 \right] + \sum_{i \ne j} A_{i,i} A_{j,j}\\
	&= \left( \trace A \right)^2 + \left(  \mathbb{E} \left[ \vert \xi_i \vert^4 \right]   -1 \right) \sum_{i=1}^{n} A^2_{i,i}  ,
	\end{align}
	where we have used that $A_{i,i} = \overline{ A_{i,i} } $ for all $ i \in \left[n\right] $ and $ \mathbb{E} \left[ \vert \xi_i \vert^2 \right]  =1$. The second summand can be computed by
	\begin{align}
	(II)&= \sum_{i \ne j, i' \ne j'} A_{i,j} \overline{ A_{i',j'}} \mathbb{E} \left[ \overline{ \xi_i} \xi_j \xi_{i'} \overline{\xi_{j'}}  \right]\\
	&= \sum_{i \ne j, i' \ne j'} A_{i,j}  A_{j',i'} \mathbb{E} \left[ \overline{ \xi_i} \xi_j \xi_{i'} \overline{\xi_{j'}}  \right]\\
	&=  \sum_{i \ne j} A^2_{i,j} \mathbb{E} \left[ \overline{ \xi_i}^2 \right] \mathbb{E} \left[ \xi^2_j  \right] + \sum_{i \ne j} A_{i,j} A_{j,i} \mathbb{E} \left[ \vert \xi_i \vert^2 \right] \mathbb{E} \left[ \vert \xi_j \vert^2  \right]\\
	&=  \sum_{i \ne j} A^2_{i,j}  \Big\vert \mathbb{E} \left[ \xi^2_i \right] \Big\vert^2  + \sum_{i \ne j} \vert A_{i,j} \vert^2\\
	&= \sum_{i \ne j} A^2_{i,j}  \vert \mathbb{E} \left[ \xi^2_i \right] \vert^2  + \sum_{i \ne j} \vert A_{i,j} \vert^2 \\
	&\overset{(a)}{=} \big\vert \mathbb{E} \left[ \xi^2_1 \right] \big\vert^2 \sum_{i \ne j} \left( \text{Re} \left(A_{i,j}\right)^2 - \text{Im} \left(A_{i,j}\right)^2   \right)   + \sum_{i \ne j} \vert A_{i,j} \vert^2\\
	&=  \left( 1 + \big\vert \mathbb{E} \left[ \xi^2_1 \right] \big\vert^2  \right) \sum_{i \ne j} \text{Re} \left( A_{i,j} \right)^2 + \left( 1- \big\vert \mathbb{E}  \left[ \xi^2_1  \right] \big\vert^2  \right) \sum_{i \ne j} \text{Im} \left( A_{i,j} \right)^2.
	\end{align}
	For equation $(a)$ we used the observation that
	\begin{equation}	
	 A^2_{i,j}+A^2_{j,i} = A^2_{i,j} + \overline{ A_{i,j}}^2  = 2 \text{Re} \left(A_{i,j}\right)^2 - 2 \text{Im} \left(A_{i,j}\right)^2.
	 \end{equation}
	 By summing up $(I)$ and $(II)$ we obtain equality (\ref{equ:secondmoment}).

\end{proof}
The lemmas above would allows us to find a lower bound for $Q_{\mathcal{Z}} \left(2u\right) $ in Lemma \ref{thm:smallballmethod}. We still need an upper bound for the Rademacher complexity $\mathbb{E} \left[ \underset{Z \in  \mathcal{Z} }{\sup}  \Big\vert \sum_{i=1}^{m}\varepsilon_i  \innerproduct{  \xi_i \xi^*_i, Z } \Big\vert   \right]$. The next lemma provides such a bound. In \cite{kueng2017low} a version of this lemma has already been presented. Nevertheless, we include a proof for completeness.
\begin{lemma}\label{lemma:PRmeanwidth} Assume that $m \ge C_1n$. Let $\alpha >0 $, $0 < \mu \le 1 $ and set $\mathcal{Z}:= \mualphaone \cap \left\{ Z \in \mathcal{S}^n: \ \Vert Z \Vert_{HS}=1  \right\}$. Then we have that
	\begin{equation}
	\mathbb{E} \left[ \underset{Z \in  \mathcal{Z} }{\sup}  \Big\vert  \innerproduct{Z , \sum_{i=1}^{m} \varepsilon_i \xi_i \xi^*_i } \Big\vert   \right] \le C_2 \left(1+\frac{1}{\alpha}\right) K^2  \sqrt{mn}.
	\end{equation}
	$C_1$ and $C_2$ are absolute constants.
\end{lemma}
\begin{proof}
	First, we note that by Hoelder's inequality and Lemma \ref{lemma:key1}  we obtain that
	\begin{equation}\label{ineq:interno1}
	\begin{split}
	\mathbb{E} \left[ \underset{Z \in  \mathcal{Z}}{\sup} \Big\vert  \innerproduct{Z , \sum_{i=1}^{m} \varepsilon_i  \xi^{\left(i\right)} \left(  \xi^{\left(i\right)}\right)^* }  \Big\vert  \right] & \le \left(   \underset{Z \in \mathcal{Z}}{\sup}  \Vert Z \Vert_1 \right)      \mathbb{E} \left[ \Big\Vert \sum_{i=1}^{m} \varepsilon_i  \xi^{\left(i\right)} \left(  \xi^{\left(i\right)}\right)^*  \Big\Vert \right]  \\
	&\le \left(1 + \frac{1}{\alpha} \right) \mathbb{E} \left[ \Big\Vert \sum_{i=1}^{m} \varepsilon_i  \xi^{\left(i\right)} \left(  \xi^{\left(i\right)}\right)^* \Big\Vert \right].
	\end{split}
	\end{equation}
	To bound $ \mathbb{E} \left[ \Big\Vert  \sum^m_{i=1} \varepsilon_i \xi^{\left(i\right)} \left(  \xi^{\left(i\right)}\right)^*   \Big\Vert  \right] $ let $ \mathcal{N} $ be an $\frac{1}{4}$-covering of the unit sphere $S^{n-1} \subset \mathbb{R}^n$ with respect to the Euclidean norm such that  
	\begin{equation}\label{ineq:unionbound}
	\vert \mathcal{N} \vert \le 12^n.
	\end{equation}
	By \cite[Lemma 5.4]{vershynin2010introduction} we have that
	\begin{equation}\label{ineq:PRepsnet}
	\Big\Vert  \sum_{i=1}^{m} \varepsilon_i  \xi^{\left(i\right)} \left(  \xi^{\left(i\right)}\right)^*   \Big\Vert \le 2 \underset{x \in \mathcal{N}}{\sup} \Big\vert  \langle  x, \sum_{i=1}^{m} \varepsilon_i  \xi^{\left(i\right)} \left(  \xi^{\left(i\right)}\right)^*  x \rangle \Big\vert.
	\end{equation}
	Fix $ x \in \mathcal{N} $ and observe that
	\begin{equation}
	\langle x, \sum_{i=1}^{m} \varepsilon_i  \xi^{\left(i\right)} \left(  \xi^{\left(i\right)}\right)^*  x \rangle = \sum_{i=1}^{m} \varepsilon_i \vert \langle  \xi^{\left(i\right)} ,x \rangle \vert^2 = \sum_{i=1}^{m} z_i,
	\end{equation}
	where we have set $z_i :=  \varepsilon_i  \vert \langle \xi_i ,x \rangle  \vert^2 $. We observe that $ \mathbb{E} \left[z_i\right] = 0 $ and, moreover,
	\begin{equation*}
	\Vert z_i  \Vert_{\psi_1} = \Vert  \langle  \xi^{\left(i\right)} ,x \rangle^2  \Vert_{\psi_1} \lesssim  \Vert  \langle  \xi^{\left(i\right)} ,x \rangle  \Vert^2_{\psi_2} \lesssim K^2 ,
	\end{equation*}
	where the first equality follows directly from the definition of the $\Vert \cdot \Vert_{\psi_1} $-norm. \footnote{For the definition of the $\Vert \cdot \Vert_{\psi_1}$-norm see, e.g., \cite[Section 5.2.4]{vershynin2010introduction}.} The first inequality can be seen using \cite[Lemma  2.7.6]{vershynin2016high} and the second one using \cite[Lemma 3.4.2]{vershynin2016high}. By the Bernstein inequality (see, e.g., \cite[Theorem 2.8.1]{vershynin2016high}) we obtain that
	\begin{equation}\label{ineq:intern5}
	\mathbb{P} \left(  \Big\vert  \langle x, \sum_{i=1}^{m} \varepsilon_i \xi^{\left(i\right)} \left(  \xi^{\left(i\right)}\right)^*  x \rangle  \Big\vert   \ge t \right) = \mathbb{P} \left(  \Big\vert \sum_{i=1}^{m} z_i \Big\vert   \ge t \right)  \le 2 \exp \left( -c \min \left\{ \frac{t^2}{mK^4}; \frac{t}{K^2} \right\}  \right) ,
	\end{equation}
	where $ c>0 $ is some numerical constant. It follows from (\ref{ineq:unionbound}), (\ref{ineq:PRepsnet}), (\ref{ineq:intern5}), and a union bound that
	\begin{align}
	\mathbb{P} \left(  \Big\Vert  \sum_{i=1}^{m} \varepsilon_i  \xi^{\left(i\right)} \left(  \xi^{\left(i\right)}\right)^*   \Big\Vert   \ge t \right) & \overset{(\ref{ineq:PRepsnet})}{\le}   	\mathbb{P} \left(   2 \underset{x \in \mathcal{N}}{\sup} \Big\vert  \langle  x, \sum_{i=1}^{m} \varepsilon_i  \xi^{\left(i\right)} \left(  \xi^{\left(i\right)}\right)^*   x \rangle \Big\vert  \ge t \right)   \\
	& \le  \sum_{x\in \mathcal{N}}    	\mathbb{P} \left(   2 \Big\vert  \langle  x, \sum_{i=1}^{m} \varepsilon_i  \xi^{\left(i\right)} \left(  \xi^{\left(i\right)}\right)^*   x \rangle \Big\vert  \ge t \right)    \\
	& \overset{(\ref{ineq:unionbound}), (\ref{ineq:intern5})}{\le} 2  \cdot 12^{n}  \exp \left( -c' \min \left\{ \frac{t^2}{mK^4}; \frac{t}{K^2} \right\}  \right)\\
	&   = 2 \exp \left( \tilde{c} n -c' \min \left\{ \frac{t^2}{mK^4} ;  \frac{t}{K^2}  \right\}   \right),
	\end{align}
	where $\tilde{c} = \log 12 $. Then, whenever $ m \ge \frac{\tilde{c}}{c} n $, we obtain that
	\begin{align}
	&\mathbb{E}  \left[  \Big\Vert  \sum_{i=1}^{m} \varepsilon_i  \xi^{\left(i\right)} \xi^*_i  \Big\Vert \right]\\
	=& \int_{0}^{\infty} \mathbb{P} \left(  \Big\Vert  \sum_{i=1}^{m} \varepsilon_i  \xi^{\left(i\right)} \left(  \xi^{\left(i\right)}\right)^*  \Big\Vert   \ge t \right)  dt \\
	\le&   \int_{0}^{   K^2 \sqrt{\frac{\tilde{c}}{c'} nm }  } 1 dt  + 2 \int_{ K^2 \sqrt{\frac{\tilde{c}}{c'} nm }   }^{ mK^2  }  \exp \left( \tilde{c} n -  \frac{c' t^2}{mK^4} \right)  dt + 2 \int_{ mK^2 }^{ \infty}   \exp \left( \tilde{c} n - c' \frac{t}{K^2}  \right)  dt\\
	\lesssim & K^2 \sqrt{nm} + \exp \left(\tilde{c} n\right) \left(  \int_{ K^2 \sqrt{\frac{\tilde{c}}{c'} nm }   }^{  \infty }  \exp \left(  -  \frac{c't^2}{mK^4} \right)  dt +    \int_{ mK^2 }^{ \infty}   \exp \left( - \frac{c't}{K^2}  \right)  dt    \right).
	\end{align}
	In order to finish we need to estimate the two integrals.  By a change of variables and  \cite[Lemma C.7]{foucart2013mathematical} we obtain that
	\begin{align}
	\int_{ K \sqrt{\frac{\tilde{c}}{c} nm }   }^{  \infty }  \exp \left(  -  \frac{c' t^2}{mK^4} \right)  dt &= \frac{\sqrt{m}K^2}{\sqrt{2c'}} \int_{\sqrt{2 \tilde{c} n}}^{\infty}  \exp \left(  \frac{-t^2}{2} \right) dt \lesssim \sqrt{m} K^2 \exp \left( -\tilde{c} n \right),\\
	\int_{mK^2 }^{ \infty}   \exp \left( - \frac{c't}{K^2}  \right)  dt   &= \frac{K^2}{c'} \exp \left( -c'm\right).
	\end{align}
	Inserting this in the inequality chain above yields that
	\begin{equation}
	\mathbb{E}  \left[  \Big\Vert  \sum_{i=1}^{m} \varepsilon_i \xi^{(i)} \left(  \xi^{\left(i\right)}\right)^*   \Big\Vert \right] \lesssim K^2 \sqrt{nm}.
	\end{equation}
	Combined with inequality (\ref{ineq:interno1}) this finishes the proof.
\end{proof}
Now we have gathered all the ingredients to complete the proof.
\begin{proof}[Proof of Lemma \ref{lemma:smallballapplied} and Lemma \ref{lemma:smallballapplied_variant}]
We will start by showing that $\mathbb{E} \left[ \vert \innerproduct{\xi \xi^*, Z }   \vert^2 \right] \gtrsim \beta \Vert Z \Vert^2_{HS} $ for all $ Z \in \mualphatwo$ in the case of Lemma \ref{lemma:smallballapplied}, or all $ Z \in \mualphatwo \cap \mathbb{R}^{n \times n} $ in the case of Lemma \ref{lemma:smallballapplied_variant}, respectively.

\noindent We first consider the second case and assume that the condition $\mathbb{E} \left[ \vert \xi_i \vert^4 \right] \ge 1+ \beta $ is satisfied for some $ \beta>0 $. By Lemma \ref{lemma:PRsecondmoments} we obtain that for all $Z \in \mathcal{M}_{2,\mu,\alpha} $ under the conditions of Lemma \ref{lemma:smallballapplied}
\begin{equation}\label{ineq:chain3}
\begin{split}
&\mathbb{E} \left[ \vert  \innerproduct{\xi \xi^*, Z }   \vert^2 \right]\\
=& \left( \trace Z\right)^2 + \left(  \mathbb{E} \vert \xi_i \vert^4   -1 \right) \sum_{i=1}^{n} Z^2_{i,i} + \left( 1 + \big\vert \mathbb{E} \left[ \xi^2_1 \right] \big\vert^2  \right) \sum_{i \ne j} \text{Re} \left( Z_{i,j} \right)^2\\
&+ \left( 1- \big\vert \mathbb{E}  \left[ \xi^2_1  \right] \big\vert^2  \right) \sum_{i \ne j} \text{Im} \left( Z_{i,j} \right)^2\\
\ge& \beta \sum_{i=1}^{n} Z^2_{i,i}   + \beta \Vert Z-\diagonal{Z} \Vert^2_{HS}\\
=& \beta \Vert Z \Vert^2_{HS}.
\end{split}
\end{equation}
Under the assumptions of Lemma \ref{lemma:smallballapplied_variant} we observe that $\sum_{i \ne j}  \text{Im} \left( Z_{i,j} \right)^2=0  $ and $\sum_{i \ne j}  \text{Re} \left( Z_{i,j} \right)^2 = \Vert Z - \diagonal{Z} \Vert^2_{HS} $. Hence, a similar argument as before also leads to
\begin{equation}
\mathbb{E} \left[ \vert  \innerproduct{\xi \xi^*, Z }   \vert^2 \right] \ge  \beta \Vert Z \Vert^2_{HS}.
\end{equation}
 Under the first assumption  we obtain by Lemma \ref{lemma:PRsecondmoments} that for all $Z \in \mualphatwo $
\begin{equation}\label{ineq:chain1}
\begin{split}
&\mathbb{E} \left[ \vert  \innerproduct{\xi \xi^*, Z }   \vert^2 \right]\\
=& \left( \trace Z\right)^2 + \left(  \mathbb{E} \vert \xi_i \vert^4   -1 \right) \sum_{i=1}^{n} Z^2_{i,i} + \left( 1 + \big\vert \mathbb{E} \left[ \xi^2_1 \right] \big\vert^2  \right) \sum_{i \ne j} \text{Re} \left( Z_{i,j} \right)^2\\
&+ \left( 1- \big\vert \mathbb{E}  \left[ \xi^2_1  \right] \big\vert^2  \right) \sum_{i \ne j} \text{Im} \left( Z_{i,j} \right)^2\\
\ge&  \beta  \sum_{i \ne j} \text{Re} \left( Z_{i,j} \right)^2 + \beta  \sum_{i \ne j} \text{Im} \left( Z_{i,j} \right)^2  \\
\ge& \beta \Vert Z - \diagonal{Z} \Vert^2_{HS}.
\end{split}
\end{equation}
Similarly, under the assumptions of Lemma \ref{lemma:smallballapplied_variant} we can again use that $\sum_{i \ne j}  \text{Im} \left( Z_{i,j} \right)^2=0  $ and $\sum_{i \ne j}  \text{Re} \left( Z_{i,j} \right)^2 = \Vert Z - \diagonal{Z} \Vert^2_{HS} $ to obtain by an analogous argument that
\begin{equation}
\mathbb{E} \left[ \vert  \innerproduct{\xi \xi^*, Z }   \vert^2 \right] \ge  \beta \Vert Z - \diagonal{Z} \Vert^2_{HS}.
\end{equation}
The remainder of the proof will be the same for Lemma \ref{lemma:smallballapplied} and Lemma \ref{lemma:smallballapplied_variant}. By Lemma \ref{lemma:key1} we have that 
\begin{align}
\Vert \diagonal{Z} \Vert_{HS} &\le \left(  \sqrt{1-\frac{1}{\left(1+ \alpha^{-1}\right)^2 }}   + 3 \mu   \right) \Vert Z \Vert_{HS}\\
&\le \left(  0.9  + \frac{1}{27}  \right) \Vert Z \Vert_{HS}\\
&\le 0.99 \Vert Z \Vert_{HS}.
\end{align}
By the triangle inequality it follows that
\begin{align}
\Vert Z -\diagonal{Z} \Vert_{HS} &\ge \Vert Z \Vert_{HS} - \Vert \diagonal{Z} \Vert_{HS}\\
&\ge \frac{1}{\tilde{C}} \Vert Z \Vert_{HS}
\end{align}
for $\tilde{C}=100$. Inserting into (\ref{ineq:chain1}) one obtains that
\begin{equation}
\mathbb{E} \left[ \vert \innerproduct{\xi \xi^*, Z }   \vert^2 \right] \ge \frac{\beta}{\tilde{C}^2} \Vert Z \Vert^2_{HS}.
\end{equation}
Hence, we have shown in all cases that $\mathbb{E} \left[ \vert  \innerproduct{\xi \xi^*, Z }   \vert^2 \right] \ge \frac{\beta}{\tilde{C}^2} \Vert Z \Vert^2_{HS} $.\\

\noindent By Lemma \ref{lemma:PRsmallballprobabilities} it follows that for all $Z\in \mualphatwo $
\begin{align}
\mathbb{P} \left( \big\vert \innerproduct{\xi \xi^*, Z }   \big\vert \ge \frac{ \sqrt{ \beta}  \Vert Z \Vert_{HS}}{\sqrt{2}\tilde{C} }  \right)  &\ge \mathbb{P} \left( \big\vert \innerproduct{\xi \xi^*, Z }   \big\vert^2 \ge \frac{ \mathbb{E} \left[ \vert \innerproduct{\xi \xi^* , Z} \vert^2 \right]  }{2  }  \right)\\
&\ge \frac{\beta^2 \Vert Z \Vert_{HS}^4}{ \tilde{C}^4 C \left( K^8 \Vert Z \Vert^4_{HS}  + \left( \trace \left( Z \right)  \right)^4  \right)}.
\end{align}
Note that for all $Z \in \mualphatwo $
\begin{align}
 \vert \trace \left(Z\right) \vert^4 \le& \Vert Z \Vert^4_1 \overset{ \left( \ref{ineq:lowrank}  \right)  }{\lesssim} \Vert Z \Vert^4_{HS},
\end{align}
where in the last inequality we also used that $ \alpha = \frac{4}{5} $. This shows that for all $Z \in \mualphatwo $ it holds that
\begin{equation}
\mathbb{P} \left( \big\vert \innerproduct{\xi \xi^*, Z }   \big\vert \ge \frac{ \sqrt{\beta} \Vert Z \Vert_{HS}}{\sqrt{2}\tilde{C} }  \right)  \gtrsim \frac{\beta^2}{K^8},
\end{equation}
where we used that $ K \gtrsim 1 $ due to (\ref{lemma:Kbound}). Now recall that $ \mathcal{Z}:= \mualphatwo  \cap  \left\{ Z\in \mathcal{S}^n: \ \Vert Z \Vert_{HS} =1  \right\}  $. Thus we have shown that
\begin{equation}
 Q_{\mathcal{Z}} \left(2u\right) \ge \frac{\beta^2}{C''K^8},
\end{equation}
where $Q_{\mathcal{Z}} \left(\cdot \right)$ is defined in (\ref{definition:Q}). We used that $ u= \frac{ \sqrt{\beta}}{2\sqrt{2} \tilde{C}}  $ and $C''>0$ is a constant chosen large enough. From Lemma \ref{lemma:PRmeanwidth} it follows that
\begin{equation}
	\mathbb{E} \left[ \underset{Z \in  \mathcal{Z} }{\sup}  \Big\vert \frac{1}{m}  \innerproduct{Z , \sum_{i=1}^{m} \varepsilon_i \xi^{\left(i\right)} \left(  \xi^{\left(i\right)}\right)^*  } \Big\vert   \right] \lesssim  K^2  \sqrt{\frac{n}{m}}.
\end{equation}
Combining this inequality with our choice of $u$ and choosing the constant in assumption (\ref{ineq:condition2}) large enough it follows that
\begin{align}
\frac{1}{u}\mathbb{E} \left[ \underset{Z \in \mathcal{Z}}{\sup} \Big\vert  \frac{1}{m} \sum_{i=1}^{m} \varepsilon_i \innerproduct{ \xi^{\left(i\right)} \left(  \xi^{\left(i\right)}\right)^*  , Z }   \Big\vert \right] \le \frac{\beta^2}{8C''K^8}.
\end{align}
Applying Lemma \ref{thm:smallballmethod} yields that with probability at least $1-2\exp\left(-2t^2\right) $
\begin{align}
\underset{Z \in \mathcal{Z}}{\inf}   \left( \frac{1}{m} \sum_{i=1}^{m} \big\vert \innerproduct{\xi^{\left(i\right)} \left(  \xi^{\left(i\right)}\right)^* , Z} \big\vert \right) &\gtrsim  \sqrt{\beta} \left(  Q_{\mathcal{Z}} \left(2u\right)  - \frac{4}{u} \mathbb{E} \left[ \underset{Z \in \mathcal{Z}}{\sup} \Big\vert  \frac{1}{m} \sum_{i=1}^{m} \varepsilon_i \innerproduct{ \xi^{\left(i\right)}   \left(  \xi^{\left(i\right)}\right)^*    , Z }   \Big\vert \right]       -\frac{t}{\sqrt{m}}   \right)\\
&\ge \sqrt{\beta} \left(  \frac{\beta^2}{2C''K^8} -  \frac{t}{\sqrt{m}}   \right).
\end{align}
Setting $t= \frac{\sqrt{m} \beta^2 }{4 C'' K^8} $ it follows that with probability at least $1- 2 \exp \left( \frac{-m\beta^4}{8 (C'')^2 K^{16}}   \right) $ it holds that
\begin{equation}
\underset{Z \in \mathcal{Z}}{\inf}   \left( \frac{1}{m} \sum_{i=1}^{m} \big\vert \innerproduct{\xi^{\left(i\right)} \left(  \xi^{\left(i\right)}\right)^* , Z} \big\vert \right) \gtrsim \frac{\beta^{5/2}}{K^8}.
\end{equation}
Hence, by the definition of $\mathcal{A}$ and $\mathcal{Z}$ it follows that
\begin{equation}\label{ineq:intern78}
\underset{Z \in \mualphatwo}{\inf} \frac{\Vert \mathcal{A} \left(Z\right) \Vert_{\ell_1}}{\Vert Z \Vert_{HS}} \gtrsim  \frac{\beta^{5/2}}{K^8} m .
\end{equation}
Due to $\alpha=\frac{4}{5} $ and Lemma \ref{lemma:key1} we have that $\Vert Z \Vert_{1} \le \frac{9}{4} \Vert Z \Vert_{HS} $ for all $Z \in \mualphatwo $. Combined with (\ref{ineq:intern78}) this shows (\ref{ineq:key22}).
\end{proof}

\section*{Acknowledgements}

This work has been supported by the German Science Foundation (DFG) in the context of the joint project {\em Bilinear Compressed Sensing} (KR 4512/2-1) as part of the Priority Program 1798 as well as the Emmy Noether Junior Research Group {\em Randomized Sensing of Signals and Images} (KR 4512/1-1). Furthermore, the authors want to thank Peter Jung for inspiring discussions.

\bibliographystyle{abbrv}
\bibliography{literature}

\end{document}